\documentclass[12pt]{article}
\usepackage[utf8]{inputenc}

\usepackage{amssymb,amsmath,amsthm}

\usepackage{url}
\usepackage{xspace}
\usepackage{enumerate}

\usepackage{tikz}
\usetikzlibrary{fit}
\usetikzlibrary{backgrounds}
\usetikzlibrary{graphs}
\usetikzlibrary{arrows}
\tikzset{
  every picture/.style={>=stealth'},
  semithick,
  nolabel/.style={label/.code={}},
  graphs/declare={gap}{__{}[draw=none,fill=none] },
  graphs/declare={gp2}{__{}[draw=none,fill=none]  -!-  __{}
    [draw=none,fill=none] },
  graphs/declare={gp3}{__{}[draw=none,fill=none]  -!-  __{}
    [draw=none,fill=none] -!-  __{}[draw=none,fill=none] },
  smallgraphs/.style={%
    graphs/every graph/.style={
      branch right=4mm,
      grow down=4mm,
      empty nodes,
      edges={line width=0.2pt},
      nodes={circle,inner sep=1.2pt}
    },
    lab/.style={draw=none,fill=none,rectangle}
  },
  every graph/.style={branch right=5mm,grow down=5mm}
}

\newcommand{\alsolabel}[2][]{
  \foreach \pos/\where/\lab in {#2} {
    \scoped[label distance=1.2mm,#1,]%
    \node[rectangle,fill=none,draw=none,%
      label={[rectangle]\where:{$%
      \pgfkeysvalueof{/mnn/label prefix}%
      \lab
      \pgfkeysvalueof{/mnn/label suffix}%
      $}},at={(\pos)}] (\lab) {};%
  };%
}
\pgfkeys{
  /mnn/label prefix/.initial={},
  /mnn/label suffix/.initial={}
}

\makeatletter
\newcommand{\mathnamenodebasic}[2][]{%
  \node[circle,fill,inner sep=2pt,#1,
  ] (#2) {};%
  \@ifundefined{mnn@savelabel}{%
    \xdef\mnn@savelabel{#2/#1}%
  }{%
    \xdef\mnn@savelabel{\mnn@savelabel,#2/#1}%
  }%
}
\newcommand{\matrixgraphfinal}[1][]{%
  \foreach \lab/\opts in \mnn@savelabel {
    \scoped[label distance=1.2mm,#1,\opts]%
    \node[\opts,rectangle,fill=none,draw=none,%
      label={[rectangle]{$%
      \pgfkeysvalueof{/mnn/label prefix}%
      \lab
      \pgfkeysvalueof{/mnn/label suffix}%
      $}},at={(\lab)}] {};%
  };%
  \global\let\mnn@savelabel\relax
}
\newcommand{\mnn@exec}[1][]{\mathnamenodebasic[#1]{\mnn@name}}
\def\mnn@gathername#1{%
  \edef\mnn@name{\mnn@name#1}\mnn@work%
}

\def\mnn@work{%
  \@ifnextchar\pgfmatrixnextcell{%
    \mnn@exec
  }{%
  \@ifnextchar\pgfmatrixendrow{%
      \mnn@exec
  }{%
      \@ifnextchar[{%
      \mnn@exec
      }{%
      \mnn@gathername%
      }%
    }%
  }%
}
\newcommand{\mathnamenode}{\def\mnn@name{}\mnn@work}
\makeatother

\makeatletter
\newcommand{\matrixgraph}{
  \begingroup
  \catcode`\&=13%
  \matrixgraph@work
}
\newcommand{\matrixgraph@work}[3][]{
  \matrix[fill=none,draw=none,rectangle,inner sep=0pt,%
    matrix anchor=north west,#1,%
    execute at begin cell=\mathnamenode] {#2};%
  \endgroup%
  \graph[use existing nodes,%
    default edge operator=complete bipartite]{#3};%
  \matrixgraphfinal[#1]%
}
\makeatother

\newcommand{\colclass}[2][0]{
\makeatletter
\xdef\clcl@verts{}
\foreach \v in {#2}{\xdef\clcl@verts{\clcl@verts(\v)}}
\begin{scope}[on background layer]
\node[fill=black!20,rotate fit=#1,fit/.expand once={\clcl@verts},inner sep=4pt,rectangle,rounded corners=8pt,draw=none] {};
\end{scope}
\makeatother
}

\theoremstyle{plain}
\newtheorem{theorem}{Theorem}[section]
\newtheorem{lemma}[theorem]{Lemma}

\theoremstyle{definition} 
\newtheorem{example}[theorem]{Example}
 \newtheorem{remark}[theorem]{Remark}

\makeatletter
\def\@gifnextchar#1#2#3{\let\@tempe#1\def\@tempa{#2}\def\@tempb{#3}%
  \futurelet\@tempc\@gifnch}

\def\@gifnch{\ifx\@tempc\@sptoken\let\@tempd\@tempb%
  \else\ifx\@tempc\@tempe\let\@tempd\@tempa\else\let\@tempd\@tempb\fi\fi\@tempd}

\def\SK@set#1{\left\{#1\right\}}
\def\SK@@set#1#2{\{#1\,:\,
    \begin{array}{@{}l@{}}#2\end{array}
\}}
\def\SK@mset#1{\left\{\!\!\left\{#1\right\}\!\!\right\}}
\def\SK@@mset#1#2{\{\!\!\{#1\,:\,
    \begin{array}{@{}l@{}}#2\end{array}
\}\!\!\}}
\def\BIG@set#1{\Big\{#1\Big\}}
\def\BIG@@set#1#2{\Big\{#1\:\Big|\:
    \begin{array}{@{}l@{}}#2\end{array}
\Big\}}
\newcommand{\Set}[1]{\@gifnextchar\bgroup{\SK@@set{#1}}{\SK@set{#1}}}
\newcommand{\Mset}[1]{\@gifnextchar\bgroup{\SK@@mset{#1}}{\SK@mset{#1}}}
\newcommand{\Bigset}[1]{\@gifnextchar\bgroup{\BIG@@set{#1}}{\BIG@set{#1}}}
\makeatother

\newcommand{\function}[2]{:#1 \rightarrow #2}
\newcommand{\Zset}{\mathbb{Z}}

\newcommand{\WL}[1]{\ensuremath{#1\text{-}\mathrm{WL}}\xspace}
\newcommand{\kwl}{\WL{k}}

\newcommand{\tw}[1]{\mathit{tw}(#1)}
\newcommand{\htw}[1]{\mathit{htw}(#1)}
\newcommand{\sub}[2]{\mathrm{sub}(#1,#2)}
\newcommand{\suub}[2]{\mathrm{sub}(#1;#2)}
\newcommand{\hhomo}[2]{\mathrm{hom}(#1,#2)}
\newcommand{\hhhomo}[2]{\mathrm{hom}(#1;#2)}
\newcommand{\inhom}[2]{\mathrm{hom}^*(#1,#2)}

\newcommand{\barz}{{\bar z}}
\newcommand{\barx}{{\bar x}}

\newcommand{\bary}{{\bar y}}

\newcommand{\alg}[4]{\mathrm{WL}_{#1}^{#2}(#3,#4)}
\newcommand{\algkstab}[2]{\mathrm{WL}_{k}(#1,#2)}

\newcommand{\algtwostab}[2]{\mathrm{WL}_{2}(#1,#2)}

\newcommand{\wlh}[2]{\mathrm{WL}_{#1,#2}}

\newcommand{\eqkwl}{\equiv_{\WL k}}

\newcommand{\clogic}{\mathsf{C}}

\newcommand{\cP}{\mathcal{P}}

\newcommand{\refeq}[1]{(\ref{eq:#1})}

\title{Local WL Invariance and\\ Hidden Shades of Regularity}
\author{Frank Fuhlbrück, Johannes Köbler, and Oleg Verbitsky%
\thanks{Supported by DFG grant KO 1053/8--1. On leave from the IAPMM, Lviv, Ukraine.}\\[4mm]
\normalsize
Institut für Informatik, Humboldt-Universität zu Berlin}

\date{}

\begin{document}

\maketitle

\begin{abstract}
The $k$-dimensional Weisfeiler-Leman algorithm is
a powerful tool in graph isomorphism testing.
For an input graph $G$, the algorithm determines a canonical coloring of 
$s$-tuples of vertices of $G$ for each $s$ between 1 and $k$.
We say that a numerical parameter of $s$-tuples is \kwl-invariant if it 
is determined by the tuple color. As an application of Dvořák's
result on \kwl-invariance of homomorphism counts, we spot 
some non-obvious regularity properties of strongly regular graphs
and related graph families. For example, if $G$ is a strongly regular graph,
then the number of paths of length 6 between vertices $x$ and $y$ in $G$
depends only on whether or not $x$ and $y$ are adjacent
(and the length 6 is here optimal). Or, the number of cycles
of length 7 passing through a vertex $x$ in $G$ is the same for every $x$
(where the length 7 is also optimal).
\end{abstract}

\section{Introduction}\label{s:intro}

The $k$-dimensional Weisfeiler-Leman algorithm (\WL k)
is a powerful combinatorial tool for detecting
non-isomorphism of two given graphs. Playing
a constantly significant role in isomorphism testing,
it was used, most prominently,  in Babai's quasipolynomial-time isomorphism algorithm~\cite{Babai16}.

For each $k$-tuple $\barx=(x_1,\ldots,x_k)$ of vertices in an input graph $G$,
the algorithm computes a canonical color $\algkstab G\barx$; see the details in Section \ref{s:prel}.
If the multisets of colors $\Mset{\algkstab G\barx}{\barx\in V(G)^k}$ and $\Mset{\algkstab H\bary}{\bary\in V(H)^k}$
are different for two graphs $G$ and $H$, then these graphs are clearly non-isomorphic,
and we say that \WL k \emph{distinguishes} them.
 If \WL k does not distinguish $G$ and $H$, we say that these graphs
are \emph{\WL k-equivalent} and write $G\eqkwl H$.
As proved by Cai, Fürer, and Immerman \cite{CaiFI92}, the \WL k-equivalence 
for any fixed dimension $k$ is strictly coarser than the isomorphism relation on graphs.
For $k=2$, an example of two non-isomorphic \WL2-equivalent graphs is provided by
any pair of non-isomorphic strongly regular graphs with the same parameters.
The smallest such pair consists of the Shrikhande and the $4\times4$ rook's graphs.
The Shrikhande graph, which will occur several times in the sequel,
is the Cayley graph of the abelian group $\Zset_4\times \Zset_4$
with connecting set $\{(1,0),(0,1),(1,1)\}$. 
Figure \ref{fig:shk} shows a natural drawing of this graph on the torus.

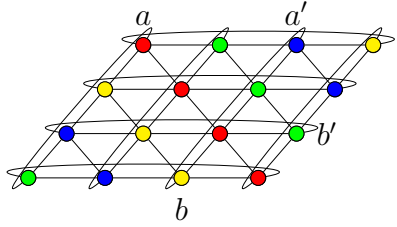
\begin{figure}[t]
\centering
\begin{tikzpicture}[
lab/.style={draw=none,fill=none,inner sep=0pt,rectangle},
every node/.style={circle,draw,inner sep=2pt},
line width=0.3pt,
wt/.style={fill=white},
vt/.style={line width=0.6pt}]
\matrixgraph[name=m1,nolabel]{
  &[3mm]&[3mm]&[3mm]&[3mm]&[3mm]&[3mm]&[3mm]&[3mm]&[3mm]\\
     &    &    & 00 &    & 01 &    & 02 &    & 03 &   \\[3.8mm]
     &    & 10 &    & 11 &    & 12 &    & 13 &    &   \\[3.8mm]
     & 20 &    & 21 &    & 22 &    & 23 &   \\[3.8mm]
  30 &    & 31 &    & 32 &    & 33 &    &   \\[3.8mm]
  }{
  00 -- 01 -- 02 -- 03 --[bend right=170] 00;
  10 -- 11 -- 12 -- 13 --[bend right=170] 10;
  20 -- 21 -- 22 -- 23 --[bend right=170] 20;
  30 -- 31 -- 32 -- 33 --[bend right=170] 30;
  00 -- 10 -- 20 -- 30 --[ bend left=170] 00;
  01 -- 11 -- 21 -- 31 --[ bend left=170] 01;
  02 -- 12 -- 22 -- 32 --[ bend left=170] 02;
  03 -- 13 -- 23 -- 33 --[ bend left=170] 03;
  00 -- 11 -- 22 -- 33;
  01 -- 12 -- 23;
  02 -- 13; 20 -- 31;
  10 -- 21 -- 32;
}

\path[every node/.style={circle,draw=black,fill=red,inner sep=2pt}] 
(00) node {} (11) node {} (22) node {} (33) node {};
\path[every node/.style={circle,draw=black,fill=green,inner sep=2pt}] 
(01) node {} (12) node {} (23) node {} (30) node {};
\path[every node/.style={circle,draw=black,fill=blue,inner sep=2pt}] 
(02) node {} (13) node {} (20) node {} (31) node {};
\path[every node/.style={circle,draw=black,fill=yellow,inner sep=2pt}] 
(03) node {} (10) node {} (21) node {} (32) node {};

\alsolabel{00/above/a,02/above/a',32/below/b,23/right/b'}

\end{tikzpicture}

\caption{The Shrikhande graph $S$ drawn on a torus; vertices of the same
  color form $4$-cycles (some  edges are not
  depicted). Though both $a$, $a'$ and $b$, $b'$ are non-adjacent, the
  common neighbors of $a$ and $a'$ are non-adjacent, while the
  common neighbors of $b$ and $b'$ are adjacent.  The automorphism
  group of $S$ acts transitively on the ordered pairs of non-adjacent
  vertices of each type~\cite{Sane15}. 
}
\label{fig:shk}
\end{figure}

Let $\hhomo FG$ denote the number of homomorphisms from a graph $F$ to a graph  $G$. 
A characterization of the \WL k-equivalence in terms of homomorphism numbers
by Dvořák \cite{Dvorak10} implies 
that the homomorphism count $\hhomo F\cdot$ is {\WL k}-invariant for each
pattern graph $F$ of treewidth at most $k$, that is, $\hhomo FG=\hhomo FH$ whenever $G\eqkwl H$. 

Let $\sub FG$ denote the number of subgraphs of $G$ isomorphic to $F$.
Lov{\'{a}}sz \cite[Section 5.2.3]{hombook} showed a close connection between 
the homomorphism and the subgraph counts, which found many applications in various context.
Curticapean, Dell, and Marx \cite{CurticapeanDM17} used this connection
to design efficient algorithms for counting the number of $F$-subgraphs in an input graph $G$.
In \cite{ArvindFKV20}, we addressed WL invariance of the subgraph counts.
Define the \emph{homomorphism-hereditary} treewidth $\htw F$ of a graph $F$ as
the maximum treewidth of the image of $F$ under an edge-surjective homomorphism.
Then Dvořák's invariance result for homomorphism counts,
combined with the Lov{\'{a}}sz relationship, implies that 
$\sub FG=\sub FH$ whenever $G\eqkwl H$ for $k=\htw F$. 

In fact, Dvořák \cite{Dvorak10} proves his result in a stronger, \emph{local} form.
To explain what we here mean by locality, we need some additional technical concepts.
A graph $F$ with $s$ designated vertices $z_1,\ldots,z_s$ is referred to as \emph{$s$-labeled}.
A tree decomposition and the treewidth of $(F,\barz)$ are defined as usually with the additional requirement
that at least one bag of the tree decomposition must contain all $z_1,\ldots,z_s$.\footnote{%
Imposing this condition is equivalent to the recursive definition given in~\cite{Dvorak10}.} 
A homomorphism from an $s$-labeled graph $(F,z_1,\ldots,z_s)$ to an $s$-labeled
graph $(G,x_1,\ldots,x_s)$ must take $z_i$ to $x_i$ for every $i\le s$.
Denote the number of such homomorphisms by $\hhhomo{F,\barz}{G,\barx}$.

The canonical coloring of the Cartesian power $V(G)^k$ produced by \WL k determines 
a canonical coloring of $V(G)^s$ for each $s$ between 1 and $k$.
Specifically, if $s<k$, we define $\algkstab G{x_1,\ldots,x_s}=\algkstab G{x_1,\ldots,x_s,\ldots,x_s}$
just by cloning the last vertex in the $s$-tuple $k-s$ times.
Dvořák \cite{Dvorak10} proves that, if an $s$-labeled graph $(F,\barz)$
has treewidth $k$, then even local homomorphism counts $\hhhomo{F,\barz}{\cdot}$
are \emph{\WL k-invariant} in the sense that $\hhhomo{F,\barz}{G,\barx}=\hhhomo{F,\barz}{H,\bary}$
whenever \kwl assigns the same color to the $s$-tuples $\barx$ and $\bary$, i.e.,
$\algkstab G{\barx}=\algkstab H{\bary}$.

Our first observation is that, like for ordinary unlabeled graphs,
this result can be extended to local subgraph counts.
Given a pattern graph $F$ with labeled vertices $z_1,\ldots,z_s$
and a host graph $G$ with labeled vertices $x_1,\ldots,x_s$, we write 
$\suub{F,\barz}{G,\barx}$ to denote the number of subgraphs $S$ of $G$ with $x_1,\ldots,x_s\in V(S)$ 
such that there is an isomorphism from $F$ to $S$ mapping $z_i$ to $x_i$ for all $i\le s$.
The local subgraph counts $\suub{F,z_1,\ldots,z_s}{\cdot}$ are \WL k-invariant
for $k=\htw{F,\barz}$, where the concept of the homomorphism-hereditary treewidth
is extended to $s$-labeled graphs in a straightforward way.
That is, not only the \WL k-equivalence
type of $G$ determines
the total number of $F$-subgraphs in $G$, but even the color $\algkstab G{x_1,\ldots,x_s}$
of each $s$-tuple of vertices determines the number of extensions of this particular tuple
to an $F$-subgraph (see Lemma~\ref{lem:main-local}).

Consider as an example the pattern graph $F=P_6$ where $P_6$ is a path through 6 vertices $z_1,\ldots,z_6$.
Consider also two host graphs $R$ and $S$ where
$R$ is the $4\times4$-rook's graph, and $S$ is the Shrikhande graph.
Since $\htw{P_6}=2$, the ``global'' invariance result in \cite{Dvorak10}
implies that $R$ and $S$ contain equally many 6-paths. Indeed,
$\sub{P_6}R=\sub{P_6}S=20448$. 

Moreover, we have $\htw{P_6,z_1,z_6}=2$; see Figure \ref{fig:p6homimg}a.
It follows that the count $\suub{P_6,z_1,z_6}{G,x,x'}$
is determined by $\algtwostab G{x,x'}$ for every graph $G$ and every pair of vertices
$x,x'$ in $G$.
If $G$ is a strongly regular graph, then
$\algtwostab G{x,x'}$ depends only on the parameters of $G$ and on whether $x$ and $x'$ are equal, adjacent, or
non-adjacent. Applied to $G\in\{R,S\}$, this justifies the fact that
$\suub{P_6,z_1,z_6}{R,x,x'}=\suub{P_6,z_1,z_6}{S,y,y'}=156$
for every pair of adjacent vertices $x,x'$ in $R$ and every pair of adjacent vertices $y,y'$ in $S$. 
If $x$ and $x'$ as well as $y$ and $y'$ are not adjacent, then
$\suub{P_6,z_1,z_6}{R,x,x'}=\suub{P_6,z_1,z_6}{S,y,y'}=180$. 

Note that the condition $\htw{P_6,z_1,z_6}=2$ is essential here.
Indeed, the slightly modified pattern $(P_6,z_2,z_5)$
does not enjoy anymore the above invariance property. For example,
for the two vertex pairs $a,a'$ and $b,b'$ in Figure \ref{fig:shk} we have
$\suub{P_6,z_2,z_5}{S,a,a'}=244$ while $\suub{P_6,z_2,z_5}{S,b,b'}=246$,
even though both pairs are non-adjacent.
The difference between the patterns $(P_6,z_1,z_6)$ and $(P_6,z_2,z_5)$,
is explained by the fact that $\htw{P_6,z_2,z_5}=3$; see Figure \ref{fig:p6homimg}b.

\begin{figure}[t]
\def\shiftall{\tikzset{xshift=9mm}}
\def\shiftsmall{\tikzset{xshift=3mm}}
\centering
\begin{tikzpicture}[
  smallgraphs,
  every node/.style={draw=gray,
  inner sep=2pt,fill=gray},
  wr/.style={fill=white,draw=black,rectangle,inner sep=1.5pt},
  gr/.style={black},
  dr/.style={rectangle,black,inner sep=1.5pt},
  bl/.style={bend left=45},
  br/.style={bend right=45},scale=1.3]
  \node[lab] {a)};
  \shiftsmall
  \graph[]{1[wr] -- 2 -- 3 ; 6[gr] -- 5 -- 4 -- 3};
  \shiftall
  \graph[empty nodes]{1[dr] -- 2 -- 3; gap -!- 5 -- 4 -- 3;5--1};
  \shiftall
  \graph[empty nodes]{1[wr] -- 2 -- 3; 6[gr] -!- gap -!- 4 -- 3;
    6 -- 1 -- 4};
  \shiftall
  \graph[empty nodes]{1[wr] -- 2 -- 3 --[bl] 1;
    6[gr] -- 5 -- 1};
  \shiftall
  \graph[empty nodes]{1[wr] -- 2; 6[gr] -- 5 -- 4 -- 1;};
  \shiftall
  \graph[empty nodes]{1[wr] -- 2 -- 3; 6[gr]--1;};
  \shiftall
  \graph[empty nodes]{1[wr] -- 2 -- 3--[bl] 1; 6[gr]--2;};
  \shiftall
  \graph[empty nodes]{1[wr] -- 2; 6[gr] -- 5 -- 2;};
  \shiftall
  \graph[empty nodes]{1[dr] -- 2 -- 3; gap -!- 5 -- 1 --[br] 3};
  \shiftall
  \graph[empty nodes]{1[dr] -- 2 -- 3; gap -!- gap -!- 4;
    3 -- 4 -- 2};
  \shiftall
  \graph[empty nodes]{1[dr] -- 2 -- 3 --[bl] 1;};
  \tikzset{yshift=-14mm,xshift=-90mm}
  \graph[empty nodes]{1[dr] -- 2 -- 3; gap -!- 5;
    3 -- 2 -- 5 -- 1};
  \shiftall
  \graph[empty nodes]{1[wr] -- 2; 6[gr] -!- gap -!- 4 -- 1 -- 6};
  \shiftall
  \graph[empty nodes]{1[wr] -- 2 -- 3; 6[gr] -- 1; 6 -- 3};
  \shiftall
  \graph[empty nodes]{1[wr] -- 2; 2 -- 6[gr] -- 1 -- 4 -- 6};
  \shiftall
  \graph[empty nodes]{1[wr] -- 2; 2 -- 6[gr] -- 1 -- 2 -- 6};
  \shiftall
  \graph[empty nodes]{1[wr] -- 2; 6[gr] -- 1};
  \shiftall
  \graph[empty nodes]{1[wr]; 6[gr] -- 1};
  \shiftall
  \graph[]{1[wr] -- 2 -- 3 ; 6[gr] -!- gap -!- 4 -- 3;
    6-- 2 -- 4;};
  \shiftall
  \graph[]{1[wr] -- 2 -- 3 ; 6[gr] -- 5 -- 2;};
  \shiftall
  \graph[]{1[wr] -- 2 -- 3 ; 6[gr] -- 3 -- 2;};
  \shiftall\shiftall
  \tikzset{yshift=14mm}
  \node[lab] {b)};
  \shiftsmall
  \graph[]{1 -- 2[wr] -- 3 ; 6 -- 5[gr] -- 4 -- 3;};
  \tikzset{yshift=-14mm}
  \graph[]{gap -!- 2[wr] -- 3 ; gap -!- 5[gr] -- 4 -- 3;
    2 -- 4; 3 -- 5};
\end{tikzpicture}
\caption{
\label{fig:p6homimg}
(a) Homomorphic images of $(P_6,z_1,z_6)$ up to isomorphism and root swapping.
(b) $(P_6,z_2,z_5)$ and its image under the homomorphism $h$ which maps
$z_1$ to $z_4$, $z_6$ to $z_3$, and fixes all other vertices.
Since $h(z_2)$ and $h(z_5)$ must be in one bag, the homomorphism-hereditary treewidth increases to~$3$.
}
\end{figure}
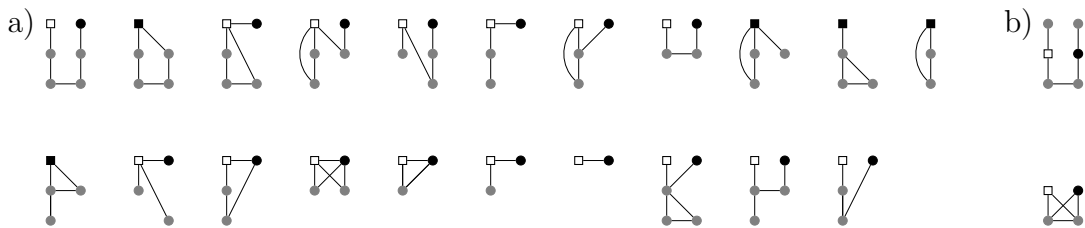

As we have seen, the number of 6-paths between two vertices is the same for any two adjacent (resp., non-adjacent)
vertices in $R$ and in $S$. The same holds true also for 7-paths and for any strongly regular graph.
This general fact follows from the \WL2-invariance of the subgraph counts for $P_s$
with labeled end vertices for $s\le7$ and from the fact that, if $G$ is a strongly regular graph,
then the color $\algkstab G{x,x'}$ is the same for all adjacent and all non-adjacent pairs $x,x'$. 
Our main purpose in this note is to collect such non-obvious regularity properties,
which are not directly implied by the definition of a strongly regular graph; see Theorem \ref{cor:srgs}.
As a further example of implicit regularity, the number of 7-cycles containing
a specified edge in a strongly regular graph does not depend on the choice of this edge.
This is actually true for a larger class of graphs, namely for constituent graphs
of association schemes, in particular, for all distance-regular graphs; see Theorem \ref{cor:const}.
Also, the number of 7-cycles passing through a specified vertex in 
a strongly regular graph does not depend on the choice of this vertex.
The last property also holds true for constituent graphs of association schemes 
and even for a yet larger class of \emph{$\wlh 12$-regular graphs} \cite{WL-recogn}; see Theorem~\ref{cor:hom12}. 

\paragraph{Related work.}
Godsil \cite{Godsil81} proves that, if $G$ is a constituent graph of an association scheme,
then the number of spanning trees containing a specified edge of $G$
is the same for every choice of the edge.
Note that this ``local'' result also has a ``global'' invariance analog.
If two graphs $G$ and $H$ are \WL2-equivalent, then they have the same
number of spanning trees. This follows from Kirchhoff's theorem because
\WL2-equivalent graphs are cospectral~\cite{Fuerer10}.

\section{Formal definitions}\label{s:prel}

\paragraph{Graph-theoretic preliminaries.}
The vertex set of a graph $G$ is denoted by $V(G)$.
A graph $G$ is \emph{$s$-regular} if every vertex of $G$ has exactly $s$ neighbors.
An $n$-vertex $s$-regular graph $G$ is \emph{strongly regular} with parameters
$(n,s,\lambda,\mu)$ if every two adjacent vertices of $G$ have $\lambda$
common neighbors and every two non-adjacent vertices have $\mu$
common neighbors. The Shrikhande and the $4\times4$ rook's graphs,
which were mentioned in Section \ref{s:intro}, have parameters $(16,6,2,2)$.

A \emph{tree decomposition} of a graph $G$ is a tree $T$ and a family
$\mathcal{B}=\{B_i\}_{i\in V(T)}$ of sets $B_i\subseteq V(G)$, called \emph{bags},
such that the union of all bags covers all $V(G)$, every edge of $G$
is contained in at least one bag, and we have $B_i\cap B_j\subseteq B_l$
whenever $l$ lies on the path from $i$ to $j$ in~$T$.
The \emph{width} of the decomposition is equal to $\max|B_i|-1$.
The \emph{treewidth} of $G$, denoted by $\tw G$, is the minimum width of a tree decomposition of~$G$.
Moreover, we define $\htw G$, the \emph{homomorphism-hereditary treewidth} of $G$, 
as the maximum $\tw {H}$ over all graphs $H$ such that there is an edge-surjective homomorphism from 
$G$ to $H$. We give a simple example for further reference.

\begin{example}\label{ex:htw-C7}
$\htw{C_s}=2$ for $3\le s\le7$, where $C_s$ denotes the cycle graph on $s$ vertices. 
Indeed, a simple inspection shows that, if $s\le7$,
then all edge-surjective homomorphic images of $C_s$ are outerplanar graphs
with one exception for the graph formed by three triangles sharing a common edge,
which is an image of $C_7$. Another argument, based on a characterization of
the class of all graphs with homomorphism-hereditary treewidth at most 2, can be found in~\cite{ArvindFKV20}.
Since this class is closed under taking subgraphs, we also have $\htw{P_s}=2$ for $s\le7$. 
\end{example}

\paragraph{The Weisfeiler-Leman algorithm.}
The original version of the algorithm, \WL2, was described by Weisfeiler and Leman in
\cite{WLe68}. For an input graph $G$,
this algorithm assigns a color to each pair of vertices $(x,y)$
and then refines the coloring of the Cartesian square $V(G)^2$ step by step. 
Initially, $\alg 20G{x,y}$ is one of three colors depending on whether $x$ and $y$
are adjacent, non-adjacent, or equal. The coloring after the $i$-th refinement step,
$\alg 2iG{x,y}$, is computed as
\begin{equation}
  \label{eq:refine-2}
  \alg 2iG{x,y}=(\alg 2{i-1}G{x,y},\Mset{(\alg 2{i-1}G{x,z},\alg 2{i-1}G{z,y})}_{z\in V(G)}),
\end{equation}
where $\Mset{}$ denotes the multiset. In words, \WL2 traces through all extensions
of each pair $(x,y)$ to a triangle $(x,z,y)$, classifies each $(x,z,y)$
according to the current colors of its sides $(x,z)$ and $(z,y)$,
and assigns a new color to each $(x,y)$ so that any two $(x,y)$ and $(x',y')$ become differently colored
if they have differently many extensions of some type.

Let $\cP^i$ denote the color partition of $V(G)^2$ after the $i$-th refinement.
Since $\cP^i$ is finer than or equal to $\cP^{i-1}$, the color partition stabilizes starting from some step
$t\le n^2$, where $n=|V(G)|$, that is, $\cP^{t+1}=\cP^t$ and, hence, $\cP^i=\cP^t$ for all $i\ge t$. 
The algorithm terminates as soon as the stabilization is reached.
We denote the final color of a vertex pair $(x,y)$ by $\algtwostab G{x,y}$.
The following simple observation plays an important role below.

\begin{remark}\label{rem:srg}
The color partition is stable from the very beginning, that is, 
$\cP^1=\cP^0$ and $\algtwostab G{x,y}=\alg 20G{x,y}$,
whenever $G$ is strongly regular, and only in this case.  
\end{remark}

The $k$-dimensional version of the algorithm, \WL k, operates similarly with $k$-tuples of vertices
and computes a stable color $\algkstab G{x_1,\dots,x_k}$ for each $k$-tuple $x_1,\dots,x_k$ of vertices of $G$.
This coloring extends to $s$-tuples $\barx=(x_1,\ldots,x_s)$ for $s<k$
and $s=k+1$ as follows. If $s<k$, we set $\algkstab G\barx=\algkstab G{x_1,\dots,x_s,\dots,x_s}$, 
i.e, extend $\barx$ to a $k$-tuple by cloning the last entry.
For $\barx=(x_1,\ldots,x_k,x_{k+1})$, we define
$\algkstab G\barx=(\algkstab G{\barx_{-1}},\ldots,\algkstab G{\barx_{-(k+1)}}$,
where $\barx_{-i}$ is obtained from $\barx$ by removing the entry $x_i$.
We do not describe \kwl for $k\ge3$ in more detail, as
the following logical characterization is sufficient for our purposes.

\paragraph{Logic with counting quantifiers.}
Cai, Fürer, and Immerman \cite{CaiFI92} established a close connection between
the Weisfeiler-Leman algorithm and first-order logic with counting quantifiers.
A \emph{counting quantifier} $\exists^m$ opens a logical formula saying that
a graph contains at least $m$ vertices with some property.
Let $\clogic^k$ denote the set of formulas in the standard first-order logical formalism
which can additionally contain counting quantifiers and use
occurrences of at most $k$ first-order variables. 

Let $G$ be a graph with $s\le k$ labeled vertices $(x_1,\ldots,x_s)$.
The \emph{$\clogic^{k}$-type} of $(G,x_1,\ldots,x_s)$ is the set of all formulas with $s$ free variables
in $\clogic^{k}$ that are true on $(G,x_1,\ldots,x_s)$.
Despite a rather abstract definition, the $\clogic^{k}$-types admit an efficient encoding
by the colors produced by the $(k-1)$-dimensional Weisfeiler-Leman algorithm.

\begin{lemma}[Cai, Fürer, and Immerman \cite{CaiFI92}]\label{lem:CFI}
Let $k\ge1$ and $0\le s\le k+1$. Then
$(G,x_1,\ldots,x_s)$ and $(H,y_1,\ldots,y_s)$ are of the same $\clogic^{k+1}$-type if and only if
$\algkstab G{x_1,\dots,x_s}=\algkstab H{y_1,\ldots,y_s}$.
\end{lemma}

For completeness, we state this result also for $k=1$, where \WL1 stands for
the classical \emph{degree refinement} routine~\cite{ImmermanL90}.

\section{Local WL invariance}
\label{sec:localWLcons}

An \emph{$s$-labeled graph} $F_\barz$ is a graph $F$ with a sequence of $s$ (not necessarily distinct)
designated vertices $\barz=(z_1,\ldots,z_s)$; cf.~\cite{LovaszS09}.
Sometimes we will also use more extensive notation $F_\barz=(F,z_1,\ldots,z_s)$,
as we already did in Sections \ref{s:intro}--\ref{s:prel}.
A homomorphism from an $s$-labeled graph $F_\barz$ to an $s$-labeled 
graph $G_\barx$ is a usual homomorphism from $F$ to $G$ taking $y_i$ to $x_i$ for every $i\le s$.

  A \emph{tree decomposition}  of an $s$-labeled graph 
  $F_\barz$ is a tree decomposition of $F$ where there is a bag 
  containing all of the labeled vertices $z_1,\ldots,z_s$.
The \emph{treewidth} $\tw {F_\barz}$ and the \emph{homomorphism-hereditary treewidth} $\htw {F_\barz}$
are defined in the same way as for unlabeled graphs.
Note that, if $F_\barz$ is a graph with $s$ pairwise distinct labeled vertices,
then $\htw {F_\barz}=k$ implies that $s\le k+1$.

\begin{remark}\label{rem:htw-rooted}
As it easily follows from the definition, for every 1-labeled graph $F_{z_1}$
we have $\tw {F_{z_1}}=\tw F$. Similarly, for a 2-labeled graph $F_{z_1,z_2}$
we have $\tw {F_{z_1,z_2}}=\tw F$ whenever the labeled vertices $z_1$ and $z_2$
are adjacent.
\end{remark}

For two $s$-labeled graphs, let $\hhomo{F_\barz}{G_\barx}$ denote the number
of homomorphisms from $F_\barz$ to $G_\barx$.
For each pattern graph $F_\barz$, the count $\sub{F_\barz}{G_\barx}$ is an invariant
of a host graph $G_\barx$. In general, a function $f$ of a labeled graph is
an \emph{invariant} if $f(G_\barx)=f(H_{\bary})$ whenever $G_\barx$ and $H_{\bary}$
are isomorphic (note that an isomorphism, like any homomorphism, must respect the labeled vertices).
We say that an invariant $f_1$ is \emph{determined} by an invariant $f_2$
if $f_2(G_\barx)=f_2(H_{\bary})$ implies $f_1(G_\barx)=f_1(H_{\bary})$.

\begin{lemma}[{Dvořák \cite[Lemma 4]{Dvorak10}}]\label{lem:dvorak}
For each $s$-labeled graph $F_\barz$, the homomorphism count
$\hhomo{F_\barz}{G_\barx}$ is determined by 
the $\clogic^{k+1}$-type of $(G,\barx)$, where $k=\tw{F_\barz}$.
\end{lemma}

Given a pattern graph $F_\barz$ with $s$ pairwise distinct labeled vertices $\barz=(z_1,\ldots,z_s)$
and a host graph $G_\barx$ with $s$ pairwise distinct labeled vertices $\barx=(x_1,\ldots,x_s)$,
let $\sub{F_\barz}{G_\barx}$ denote the number of subgraphs $S$ of $G$ with $x_1,\ldots,x_s\in V(S)$ such that
there is an isomorphism from $F$ to $S$ mapping $z_i$ to $x_i$ for all $i\le s$.
Using the well-known inductive approach due to Lov{\'{a}}sz~\cite{Lovasz71}
to relate homomorphism and subgraph counts, from Lemma \ref{lem:dvorak}
we derive a related fact about local invariance of subgraph counts.

\begin{lemma}\label{lem:main-local}
For each $s$-labeled pattern graph $F_\barz$, the subgraph count
$\sub{F_\barz}{G_\barx}$ is determined by the $\clogic^{k+1}$-type of $(G,\barx)$, where $k=\htw{F_\bary}$.
\end{lemma}

\begin{proof}
Let $\inhom{F_\barz}{G_\barx}$ denote the number of
\emph{injective} homomorphisms from $F_\barz$ to $G_\barx$.
Counting injective homomorphisms and subgraphs is essentially equivalent, since
\[
\sub{F_\barz}{G_\barx}=
\frac{\inhom{F_\barz}{G_\barx}}{|\mathrm{Aut}(F_\barz)|},
\]
where $\mathrm{Aut}(F_\barz)$ is the automorphism group of 
$F_\barz$ (as any homomorphism, automorphisms have to respect the labeled vertices).
Therefore, it is enough to prove that
the injective homomorphism count $\inhom{F_\barz}{G_\barx}$
is determined by the $\clogic^{k+1}$-type of $(G,\barx)$.

We use induction on the number of vertices in $F$.
The base case, namely $V(F)=\{z_1\}$, is obvious.
Suppose that $F$ has at least two vertices.

Call a partition $\alpha$ of $V(F)$ \emph{proper}
if every element of $\alpha$ is an independent set of vertices in $F$.
The partition of $V(F)$ into singletons is called \emph{discrete}.
A proper partition $\alpha$ determines the homomorphic image $F/\alpha$ of $F$
where $V(F/\alpha)=\alpha$ and two elements $A$ and $B$ of $\alpha$
are adjacent if between the vertex sets $A$ and $B$ there is at least one edge in $F$.
Let $\barz/\alpha=(A_1,\dots,A_s)$ where $A_i$ is the element of the partition $\alpha$ containing $z_i$.
The $s$-labeled graph $F_\barz/\alpha$ is defined by $F_\barz/\alpha=(F/\alpha)_{\barz/\alpha}$.
Note that
$$
\hhomo{F_\barz}{G_\barx}=
\inhom{F_\barz}{G_\barx}+
\sum_\alpha\inhom{F_\barz/\alpha}{G_\barx}  
$$
where the sum goes over the non-discrete proper partitions of $V(F)$.  
Note also that the number of labeled vertices in each $F_\barz/\alpha$ stays $s$, 
but some of them may become equal to each other.
In this case, we have $\inhom{F_\barz/\alpha}{G_\barx}=0$,
and all such $\alpha$ can be excluded from consideration.

Now, $\hhomo{F_\barz}{G_\barx}$ is determined by the $\clogic^{k+1}$-type of $(G,\barx)$
by Lemma \ref{lem:dvorak} because $\tw {F_\barz}\leq\htw {F_\barz}$.
Each $\inhom{F_\barz/\alpha}{G_\barx}$ is determined by the $\clogic^{k+1}$-type of $(G,\barx)$
by the inductive assumption. It follows that $\inhom{F_\barz}{G_\barx}$ is also determined.
\end{proof}

\section{Implicit regularity properties}

In what follows, by \emph{$s$-path} (resp.\ \emph{$s$-cycle}) we mean a path graph $P_s$
(resp.\ a cycle graph $C_s$) on $s$ vertices.

\begin{theorem}\label{cor:srgs}\hfill
  \begin{enumerate}[\bf 1.]
  \item For each $s\le7$, the number of $s$-paths between distinct
    vertices $x$ and $y$ in a strongly regular graph $G$ depends only
    on the adjacency of $x$ and $y$ (and on the parameters of~$G$).
\item For each $3\le s\le5$, the number of $s$-cycles containing two
  distinct vertices $x$ and $y$ of a strongly regular graph $G$
  depends only on the adjacency of $x$ and $y$ (and on the parameters
  of~$G$).
  \end{enumerate}
\end{theorem}

\begin{proof}
Let $s\le7$. Suppose that the path $P_s$ goes through vertices $z_1,\ldots,z_s$ in this order
and, similarly, the cycle $C_s$ is formed by vertices $z_1,\ldots,z_s$ in this cyclic order.
  For cycles, $\htw{C_s,z_1,z_s}=\htw{C_s}=2$, where the first
  equality is due to Remark \ref{rem:htw-rooted} 
  and the second equality is noticed in
  Example \ref{ex:htw-C7}. 
Let $h$ be an edge-surjective homomorphism from a 2-labeled path $(P_s,z_1,z_s)$.
If $h(z_1)=h(z_s)$, then the image of $(P_s,z_1,z_s)$ under $h$ is also a homomorphic image of $(C_{s-1},z_1,z_1)$.
If $h(z_1)\ne h(z_s)$, then adding an edge between $h(z_1)$ and $h(z_s)$ does not increase
the treewidth of the 2-labeled image graph, which can then be seen as a homomorphic image of $(C_s,z_1,z_s)$.
Therefore, $\htw{P_s,z_1,z_s}\le\max(\htw{C_s,z_1,z_s},\htw{C_{s-1},z_1})$.
This implies $\htw{P_s,z_1,z_s}\le2$.  A
  straightforward inspection shows also that $\htw{C_s,z_1,z_i}=2$ for
  all $i\le s$ if $s\le5$.  
We conclude by Lemma \ref{lem:main-local} that $\suub{P_s,z_1,z_s}{G,x,y}$ for $s\le7$
and $\suub{C_s,z_1,z_i}{G,x,y}$ for $s\le5$ are determined by the $\clogic^3$-type of $(G,x,y)$
and hence, by Lemma \ref{lem:CFI}, by the canonical color $\algtwostab G{x,y}$.
As noticed in Remark \ref{rem:srg}, the color is determined by the adjacency of $x$ and~$y$.
\end{proof}

The next result applies to a larger class of graphs.
The concept of an \emph{association scheme} appeared in statistics (Bose \cite{BoseM59})
and, also known as a \emph{homogeneous coherent configuration} (Higman~\cite{Higman64}),
plays an important role in algebra and combinatorics \cite[Chapter 17]{CameronvL91}.
The standard definition and an overview of many important results in this area
can be found in \cite{Ponomarenko-book,Zieschang05}.
For our purposes, an association scheme can be seen as a complete directed graph
with colored edges such that
\begin{itemize}
\item 
all loops $(x,x)$ form a separate color class,
\item 
if two edges $(x,y)$ and $(x',y')$ are equally colored, then
their transposes $(y,x)$ and $(y',x')$ are equally colored too, and
\item 
the coloring is stable under the \WL2 refinement.
\end{itemize}
Each non-loop color class, seen as an undirected graph, forms
a \emph{constituent graph} of the association scheme;
this concept first appeared apparently in~\cite{Godsil81}.

A strongly regular graph and its complement can be seen as the two constituent
graphs of an association scheme. Moreover, all distance-regular graphs \cite{BrouwerCN89}
are constituent graphs (in fact, every connected strongly regular graph is distance-regular).

\begin{lemma}\label{lem:as}
  Let $\alpha\function{V^2}C$ be a coloring defining an association scheme,
and $G$ be a constituent graph of this scheme. Then
$\algtwostab G{x,y}=\algtwostab G{x',y'}$ whenever $\alpha(x,y)=\alpha(x',y')$.
\end{lemma}

\begin{proof}
  Using the induction on $i$, we will prove that the equality $\alpha(x,y)=\alpha(x',y')$
implies the equality $\alg 2iG{x,y}=\alg 2iG{x',y'}$ for every $i$.
In the base case of $i=0$ this is true by the definition of a constituent graph.
Suppose that the claim is true for some $i$ for all $x,y,x',y'\in V$.

To prove the claim for $i+1$, fix $x,y,x',y'$ such that $\alpha(x,y)=\alpha(x',y')$.
By the induction assumption, $\alg 2iG{x,y}=\alg 2iG{x',y'}$.
For a pair of colors $p\in C^2$, let $T(p)=\Set{z}{(\alpha(x,z),\alpha(z,y))=p}$
and $T'(p)=\Set{z'}{(\alpha(x',z'),\alpha(z',y'))=p}$. 
Since $\alpha$ defines an association scheme, this coloring is not
refinable by \WL2. This means that
\begin{equation}
  \label{eq:Tp}
  |T(p)|=|T'(p)|\text{ for every }p\in C^2.
\end{equation}
By the induction assumption, for every $z\in T(p)$ and $z'\in T'(p)$ we have
$$
(\alg 2iG{x,z},\alg 2iG{z,y})=(\alg 2iG{x',z'},\alg 2iG{z',y'}).
$$
Along with Equality \refeq{Tp}, this implies that
$$
\Mset{(\alg 2iG{x,z},\alg 2iG{z,y})}_{z\in V}=\Mset{(\alg 2iG{x',z'},\alg 2iG{z',y'})}_{z'\in V}
$$
and, therefore, $\alg 2{i+1}G{x,y}=\alg 2{i+1}G{x',y'}$, as desired.
\end{proof}

\begin{theorem}\label{cor:const}
Let $G$ be a constituent graph of an association scheme.
If $3\le s\le7$, then the number of $s$-cycles containing an edge $xy$
in $G$ does not depend on the choice of~$xy$.
\end{theorem}

\begin{proof}
As already noted in the proof of Theorem \ref{cor:srgs}, $\htw{C_s,z_1,z_2}=2$,
where two labeled vertices $z_1,z_2$ are consecutive in $C_s$.
Like in the proof of Theorem \ref{cor:srgs}, we conclude that the count
$\suub{C_s,z_1,z_2}{G,x,y}$ is determined by the canonical color $\algtwostab G{x,y}$.
Let $x'y'$ be another edge of $G$. By Lemma \ref{lem:as},
$\algtwostab G{x,y}=\algtwostab G{x',y'}$ or $\algtwostab G{x,y}=\algtwostab G{y',x'}$.
It remains to notice that 
$\suub{C_s,z_1,\allowbreak z_2}{G,x',y'}=\suub{C_s,z_1,z_2}{G,y',x'}$
just because $C_s$ has an automorphism transposing $z_1$ and~$z_2$.
\end{proof}

The optimality of Theorems \ref{cor:srgs} and \ref{cor:const} is certified by Table~\ref{table}.

\begin{table}[t]
\centering
$\begin{array}{|c|c|r|}\hline
F_\bary & G_\barx &\suub{F_\bary}{G_\barx}\\
\hline
(P_8,x_1,x_8) & (S,a,a') & 2500\\
(P_8,x_1,x_8) & (S,b,b') & 2522\\
(C_6,x_1,x_3) & (S,a,a') & 72\\
(C_6,x_1,x_3) & (S,b,b') & 74\\
(C_6,x_1,x_4) & (\allowbreak S,a,a') & 92\\
(C_6,x_1,x_4) & (S,b,b') & 94\\
\hline
(C_8,x_1,x_2) & (\overline {S},a,a') & 48832\\
(C_8,x_1,x_2) & (\overline {S},b,b') & 48788\\
\hline
\end{array}$
\caption{$S$ and $\overline{S}$ denote the Shrikhande graph and
its complement, respectively;  $a$, $a'$, $b$, and $b'$ are the four vertices in $S$
shown in Figure~\ref{fig:shk}.}
\label{table}
\end{table}

Theorem \ref{cor:hom12} below applies to a yet larger class of graphs.
Following \cite{WL-recogn}, we call a graph $G$ \emph{$\wlh 12$-regular} if
$\algtwostab G{x,x}=\algtwostab G{x',x'}$ for every two vertices $x$ and $x'$ of $G$, 
that is, \WL2 determines a monochromatic coloring of $V(G)$.
Applying Lemma \ref{lem:as} in the case $y=x$ and $y'=x'$, we see that 
any constituent graph of an association scheme is $\wlh 12$-regular.
More $\wlh 12$-regular graphs can be obtained by observing that
this graph class is closed under taking graph complements and that,
if $G_1$ and $G_2$ are two $\wlh 12$-regular \WL2-equivalent graphs,
then the disjoint union of $G_1$ and $G_2$ is also $\wlh 12$-regular.
In terms of the theory of coherent configurations \cite{Ponomarenko-book},
a graph $G$ is $\wlh 12$-regular if and only if
the \emph{coherent closure} of $G$ is an association scheme.

\begin{theorem}\label{cor:hom12}\hfill
  \begin{enumerate}[\bf 1.]
  \item 
For each $s\le7$, the number of $s$-paths emanating from a vertex $x$
in a $\wlh 12$-regular graph $G$ is the same for every~$x$.
\item 
For each $3\le s\le7$, the number of $s$-cycles containing a vertex $x$
in a $\wlh 12$-regular graph $G$ is the same for every~$x$.
  \end{enumerate}
\end{theorem}

\begin{proof}
Taking into account Remark \ref{rem:htw-rooted} and Example \ref{ex:htw-C7}, we obtain
$\htw{P_s,z_1}=\htw{P_s}\le2$ and $\htw{C_s,z_1}=\htw{C_s}=2$.
Like in the proofs of Theorems \ref{cor:srgs} and \ref{cor:const}, we conclude that the counts
$\suub{P_s,z_1}{G,x}$ and $\suub{C_s,z_1}{G,x}$ are determined by the canonical color $\algtwostab G{x,x}$.
By the definition of $\wlh 12$-regularity, this color
is the same for all vertices of~$G$.
\end{proof}

Theorem~\ref{cor:hom12} is optimal regarding the restriction $s\le7$.
Indeed, consider the strongly regular graphs with parameters $(25,12,5,6)$ without
nontrivial automorphisms.  These are two complementary graphs.
Specifically, we pick $H=P_{25.02}$ (with graph6 code
\verb|X}rU\adeSetTjKWNJEYNR]PL| \verb|jPBgUGVTkK^YKbipMcxbk`{DlXF|)
in the graph database \cite{BrouwerPaulus}, which is based in this part on \cite{Paulus73}.
The graph $H$ has a vertex $x$ with $\sub{P_8,z_1}{H,x}=11115444$
and $\sub{C_8,z_1}{H,x}=5201448$ and a vertex $y$ with
$\sub{P_8,z_1}{H,x}=11115510$ and $\sub{C_8,z_1}{H,x}=5201580$.


\end{document}